\documentclass[12pt]{article}
\usepackage{pudlatex}
\usepackage{bussproofs}

\newcommand\Res{\mathrm{Res}}
\newcommand\SF{\mathrm{SF}}
\newcommand\EF{\mathrm{EF}}
\newcommand\PV{\mathrm{PV}}
\newcommand\imp{\mathop{\rm Imp}\nolimits}
\newcommand\nimp{\mathop{\rm NImp}\nolimits}
\newcommand\impres{\imp_{\Res}}
\newcommand\nimpres{\nimp_{\Res}}

\newcommand\ang[1]{\langle #1 \rangle}

\newcommand\seq{\longrightarrow}
\newcommand\PP{\mathrm{P}}

\def\Krajicek{Kraj\'{i}\v{c}ek}

\title{Quantified propositional calculi \\ and narrow implicit proofs} 
\author{Pavel Pudl\'ak\thanks{
Institute of Mathematics, Czech Academy of Sciences,
\texttt{pudlak@math.cas.cz}, \texttt{thapen@math.cas.cz};
supported by the
Czech Academy of Sciences (RVO 67985840)} \thanks{Partially supported by GA \v{C}R grant  25-16311S.} 
~ and~ Neil Thapen\footnotemark[1] }

\begin{document}

\maketitle

\begin{abstract}
In the implicit version of a propositional proof system $Q$, we work with $Q$-proofs
that are not written down directly, but are succinctly encoded by circuits.
Thus implicit $Q$-proofs are potentially exponentially shorter than 
usual $Q$-proofs. We study  \emph{narrow} implicit proofs, a restricted version of this notion,
in which lines in the encoded proof can only have polynomial size. 
We use a cut-elimination construction to show that
$G_{i+1}$ is equivalent to narrow implicit $G_i$,
 for $i \ge 1$, where $G_i$ is the 
extension of Frege allowing reasoning with $\Sigma^q_i$
 quantified propositional formulas.
We show that $G_1$ is equivalent to implicit resolution.
\end{abstract}

%%%%%%%%%%%%%%%%%%%%%%%%%%%%%%%%%%%%%%%%%%%%%%%%%%%%%%%
%%%%%%%%%%%%%%%%%%%%%%%%%%%%%%%%%%%%%%%%%%%%%%%%%%%%%%%
\section{Introduction} \label{sec:intro}
%%%%%%%%%%%%%%%%%%%%%%%%%%%%%%%%%%%%%%%%%%%%%%%%%%%%%%%
%%%%%%%%%%%%%%%%%%%%%%%%%%%%%%%%%%%%%%%%%%%%%%%%%%%%%%%

Kraj\'i\v cek~\cite{krajicek-implicit} defined an operation that from two proof systems $P$ and $Q$ produces another proof system, which he denoted by $[P,Q]$ and we will denote by $\imp_P(Q)$.
%\footnote{We apologize for changing notation, but brackets are used in mathematics too often. Our notation emphasizes that, at least in the cases we consider, the choice of~$Q$ seems more important for the strength of the resulting system 
%than the choice of~$P$.} 
A proof of a formula $\phi$ in $\imp_P(Q)$ is a pair $(\Pi,C)$, where $C$ is a Boolean circuit that succinctly defines a possibly exponential-size $Q$-proof of $\phi$, and $\Pi$ is a $P$-proof of a formula that expresses this fact. 

In the special case when $P=Q$, Kraj\'i\v cek calls such a proof system \emph{implicit $P$} and denotes it by $iP$. This operation seems to produce from $P$ an essentially stronger proof system, at least for some natural  systems. We will not define this general concept more precisely, because we will only be interested in a restricted version.% of the operation~$\imp(P,Q)$.) 

%The concept of an implicit proof is very natural, but unfortunately, it is not well-behaved. For instance,
%it is shown in~\cite{khaniki} that
% there exist proof systems $P$, $Q$ and $R$ such that
% $R$ polynomially simulates $Q$, but 
%$\imp_P(R)$ does not polynomially simulate $\imp_P(Q)$. 
%In that example one of the proof systems is rather unnatural; 
%it is possible that for natural proof systems that are defined by sequences of formulas 
%this cannot happen.

In this paper we will only consider implicit proof systems constructed from proof systems
of a certain natural form,
in which proofs are sequences of lines, and each line can be represented by a formula.
Furthermore we impose the restriction that, although our circuit-defined proofs may be exponentially large,
each line can only have polynomial size. We call these \emph{narrow implicit} proof systems
and use the notation~$\nimp_P(Q)$.\footnote{
We choose this notation to emphasize the different roles played by~$P$ and~$Q$ and  that, at least in the cases we consider, the choice of~$Q$ seems more important for the strength of the resulting system 
than the choice of~$P$.} 
This modification was also defined by Kraj\'i\v cek in~\cite{krajicek-implicit}.

The concept of an implicit proof is very natural, but unfortunately, it is not well-behaved. For instance,
it is shown in~\cite[Proposition~7.1]{khaniki} that, under the assumption that
there exists $f$ computable in time $2^{O(n)}$ with $f \notin_{i.o.}$ P/poly,
we can define proof systems $P$ and $Q$  such that $P \equiv_p Q$
 but $\impres(P) \not \equiv_p \impres(Q)$.
An example of bad behaviour for narrow implicit proofs follows from our results, 
namely that $G_1 \ge_p \EF$, but it is unlikely that 
$\nimpres(G_1) \ge_p \nimpres(\EF)$, since this would imply
$G_2 \ge_p G$.
This is because $\nimpres(G_1) \equiv_p G_2$  by our Theorem~\ref{the:main},
 and $\nimpres(\EF) \ge_p G$ by~\cite{krajicek-implicit}. (See below for relevant definitions.)

%Our main aim is to show a connection between these systems and proofs obtained by cut-elimination.
% The idea is that by cut-elimination one gets an exponentially larger proof when one eliminates the top level of cuts, but the resulting proof is still easily describable. 
% This is based on a trick that first appeared in~\cite{buss-insitu}, which is to construct a dag-like proof
%  with top-level cuts eliminated instead of the traditionally used tree-like proofs. 
%  This enables one to keep the size of the top-level-cut-free proofs only single-exponential size.
%   We will demonstrate this idea with quantified propositional calculi. 
% Work on cut-elimination with polynomial-time machinery also appeared in~\cite{ab_10, bb_10}.

Our main aim is to show a connection between implicit proof systems and proofs obtained by cut-elimination.
 The idea is that by cut-elimination one gets an exponentially larger proof when one eliminates the top level of cuts, but the resulting proof is still easily describable. 
 This is based on a trick that first appeared in~\cite{buss-insitu}, 
 which is to apply cut elimination to a treelike proof, using operations that are in some sense local,
 in a way that results in a daglike proof with only a small increase in size.
This enables one to eliminate top-level cuts from a daglike proof with only a single-exponential
increase in size, by first unfolding the daglike proof
into a treelike one (a single-exponential size increase) and then doing cut elimination as described.
    We will apply this idea to quantified propositional calculi. 
 Work on cut-elimination with polynomial-time machinery also appeared in~\cite{ab_10, bb_10}.

The basic proof system for quantified propositional tautologies is a sequent calculus,
 with the quantifier rules adapted to propositional quantifiers,
 denoted by~$G$. Its fragments, in which the quantifier complexity of formulas is restricted to be $\Sigma^q_i$, are denoted by~$G_i$~\cite{krajicek-pudlak90}. See Section~\ref{sec:systems} below for definitions. Our main result is
\begin{theorem} \label{the:main} 
$G_{i+1} \equiv_p \nimpres(G_i)$ for $i \ge 1$, and
$G_1 \equiv_p \nimpres(\Res)$. 
\end{theorem}
Theorem~\ref{the:main} follows immediately from Lemma~\ref{lem:G_i+1_simulates} and 
Theorems~\ref{the:cutelim_theorem} and~\ref{the:impres} below.
An immediate corollary is that also $G_1 \equiv_p \impres(\Res)$, 
since any implicit resolution proof is automatically
narrow, as its clauses cannot be larger than the formula being proved~\cite{krajicek-implicit}.
Here  a system $P$ polynomially simulates a system~$Q$, or $P \ge_p Q$,
if we can convert any $Q$-proof of a (quantifier-free) propositional tautology into a $P$-proof
of the same tautology in polynomial time. They are polynomially equivalent,
$P \equiv_p Q$, if both directions hold.

One consequence of the theorem is that we now have further examples
of implicit proofs producing stronger proof systems, assuming quite a plausible hypothesis about the increasing hierarchy of fragments of~$G$. In fact we show that even \emph{narrow} implicit proof systems are stronger assuming this hypothesis.

The result that $G_1 \equiv_p \impres(\Res)$ has also recently been shown in independent
work by Fleming, Grosser, Pitassi and Robere~\cite{FGPR}, using an approach broadly similar to ours.

\bigskip

The rest of this section is an outline of the paper.
In Section~\ref{sec:implicit} we introduce narrow implicit proofs and some of their basic properties. We prove
\begin{lemma} \label{lem:G_i+1_simulates}
$G_{i+1} \ge_p \nimpres(G_i)$ for  $i \ge 0$.
\end{lemma}
This gives the easy direction of both items of Theorem~\ref{the:main},
since clearly $\nimpres(G_0) \ge_p \nimpres(\Res)$. The lemma is proved by 
observing that the standard soundness proof for succinctly-described
exponential-size $G_i$ proofs can be formalized in the bounded arithmetic theory~$T^{i+1}_2$,
which gives a simulation by $G_{i+1}$ by results of~\cite{krajicek-pudlak90}.
At the end of the section we introduce how we will also use bounded arithmetic
for the other direction of the simulation, in which we must construct implicit proofs.

In Section~\ref{sec:unfolding} we discuss a basic step in cut-elimination, of unfolding a daglike proof into 
an exponentially-large treelike proof, and show how  the resulting proof can be ``implicit'' in our sense.

Section~\ref{sec:main} contains the main construction of the paper. We show
\begin{theorem} \label{the:cutelim_theorem}
$G_{i+1} \le_p \nimpres(G_i)$ for  $i \ge 1$,
and 
$G_1 \le_p \nimpres(SG_0)$. 
\end{theorem}
Here $SG_0$ is $G_0$ plus the substitution rule, and we can think of it as a formulation
of the standard \emph{substitution Frege} system (see Section~\ref{sec:Frege}).
The proof of Theorem~\ref{the:cutelim_theorem} is based on cut-elimination. To this end one has to use a modification of the standard cut-elimination algorithm in order to obtain only exponential blow-up. Our construction is similar to the one used by Buss~\cite{buss-insitu}, but we need to take extra care in order to prove that the exponential-size proof has a succinct representation using a polynomial-size circuit,
 and that the soundness of this representation is provable in resolution. 
We prove the~$i{=}0$ case first, and then show how it generalizes to~$i{\ge}1$ and that
in this case one does not need to add the substitution rule.

In Section~\ref{sec:impres} we strengthen the result about $G_1$ in Theorem~\ref{the:cutelim_theorem} to
\begin{theorem} \label{the:impres}
$G_1 \le_p \nimpres(\Res)$.
\end{theorem}
The proof is independent of Theorem~\ref{the:cutelim_theorem}
and does not (directly) use cut-elimination, but rather
known connections between $G_1$, the bounded arithmetic theory $T^1_2$, and TFNP search
problems in the class PLS.
The proof constructed here perhaps sheds some light on the nature of implicit proofs,
since it has a structure that does not resemble any usual concrete proof.

At the end of the section we suggest a generalization of Theorem~\ref{the:impres} 
to~$G_i$ for~$i \ge 1$, different from the one in Theorem~\ref{the:main}
in that it is about the complexity of the definition of the proofs, rather than about the 
expressivity of the lines.
Ideally, we would like to claim that $G_{i+1}$ is equivalent to 
a system one might call $\Sigma^p_i$-$\nimp(\Res)$,  a generalization of implicit resolution in which
we equip the circuit defining the large resolution refutation with gates that make $\Sigma^p_i$ oracle calls.
We prove a technical result in this direction, but it is not clear how to define such a system naturally.

Some notation: we write $[a]$ for the interval $[0, a{-}1]$.

%%%%%%%%%%%%%%%%%%%%%%%%%%%%%%%%%%%%%%%%%%%%%%%%%%%%%%%
%%%%%%%%%%%%%%%%%%%%%%%%%%%%%%%%%%%%%%%%%%%%%%%%%%%%%%%
\section{Proof systems} \label{sec:systems}
%%%%%%%%%%%%%%%%%%%%%%%%%%%%%%%%%%%%%%%%%%%%%%%%%%%%%%%
%%%%%%%%%%%%%%%%%%%%%%%%%%%%%%%%%%%%%%%%%%%%%%%%%%%%%%%

In general our proofs will be in versions of Gentzen's sequent calculus $LK$, as presented 
for example in~\cite{buss-handbook}, adapted for particular classes of formulas.
We need to be especially careful with the formalizations, as some details, which normally do not matter,
may play an important role when we use succinctly presented proofs. Therefore we have to say  precisely which 
calculus we will use.
In Section~\ref{sec:main} we will also introduce some auxiliary rules, which will not appear in the 
final constructed proof but which will make the process of implicit cut-elimination smoother, 
by for example combining some rules together to postpone having to add extra vertices to a proof tree.

\subsection{Resolution}

A resolution refutation of a CNF $F$ is a sequence of clauses, where each clause
uses only variables from $F$ and has no repeated variables.
The sequence starts with the clauses of $F$ and ends with the empty clause.
Each clause is derived from
earlier clauses either by weakening or by the resolution rule
that derives $C \vee D$ from $C \vee p$ and $D \vee \neg p$,
where we do not allow clashing literals in $C$ and $D$ and we remove repeated literals
from $C \vee D$.

\subsection{The systems $G$ and $G_i$}

We describe the system~$G$ for reasoning with quantified Boolean formulas.
This is based on  the sequent calculus $LK$, using essentially all rules, including quantifier rules. 
However we have to adapt how the $\exists$-right and $\forall$-left rules work, since these usually introduce
a bound variable to replace a first-order term, and in our propositional calculus there are no terms.
In this paper we will use quantifier-free formulas in place of terms in~$LK$ in the quantifier rules; 
see the definitions below. 
Systems with quantifier rules in which quantified propositions are allowed have also been considered 
\cite{krajicek-pudlak90} and these are polynomially equivalent to the system with quantifier-free formulas.
Furthermore, it has been shown that it suffices to use the truth constants instead of general formulas, 
see~\cite{morioka,cook-nguyen}.

\sloppypar 
We distinguish between bound variables $x,y,z\dots$ and free variables $a,b,c\dots$.
Cedents are sequences of formulas which may contain repetitions. 
We think of them as sequences rather than sets because we will need to distinguish
between different occurrences of a formula in a cedent; for example
it may occur both as a principal formula and a side formula.
Rather than having
an explicit exchange rule, we allow these sequences to be arbitrarily reordered inbetween rules
(so effectively a cedent is a multiset).

We specify a technical restriction on the syntax of quantified formulas: in a well-formed formula,
for each bound variable~$x$, quantifiers over~$x$
cannot occur inside the scope of another quantifier over~$x$.

The axioms and rules of $G$ are as follows. See e.g.~\cite{buss-handbook} for further details.
\bi
\item 
Axioms are all sequents of the form $\Gamma,A\seq\Delta,A$.
\item 
We have the weakening and contraction rules.
A special case of  weakening is the \emph{repetition} rule, where we derive a sequent from
itself with no changes. We will use this as a  padding to make large proofs uniform.
\item 
We have the cut rule, and the standard left and right introduction rules for the propositional
connectives $\bigwedge$, $\bigvee$, $\neg$.
\item
The quantifier rules are
\[
\exists\textrm{-left: } 
\frac{\Gamma, \phi(a) \seq \Delta}{\Gamma, \exists x  \phi(x) \seq \Delta}
\qquad
\exists\textrm{-right: } 
\frac{\Gamma \seq \Delta, \phi(a/\beta)}{\Gamma \seq \Delta, \exists x  \phi(x)}
\]
\[
\forall\textrm{-left: } 
\frac{\Gamma, \phi(a/\beta) \seq \Delta}{\Gamma, \forall x  \phi(x) \seq \Delta}
\qquad
\forall\textrm{-right: } 
\frac{\Gamma \seq \Delta, \phi(a)}{\Gamma \seq \Delta, \forall x  \phi(x)}
\]
where $\beta$ is a quantifier-free formula, and where in the 
$\exists$-left and $\forall$-right rules the free variable $a$ does not occur in $\Gamma$ or $\Delta$
(the \emph{eigenvariable condition}).
The notation $\phi(a/\beta)$ means we
substitute every occurrence of~$a$ in~$\phi$ with $\beta$.
Thus by these rules we can replace some,
but not necessarily all, instances of $\beta$ with the bound variable~$x$.
\ei

A $G$ proof of a quantifier-free formula $\phi$ is a proof in the above sequent calculus
ending in the sequent $\seq \phi$.
A $G_i$ proof is a $G$ proof in which all formulas are in $\Sigma^q_i$.
We will also assume that in $G_i$ proofs all formulas are in \emph{weakly prenex form}, 
which means that they are either in $\Sigma_{i-1}^q\cup\Pi_{i-1}^q$, 
or of the form $\exists x_1\dots\exists x_k\phi$ with $\phi\in\Pi_{i-1}^q$.

\subsection{Frege and its extensions} \label{sec:Frege}

A \emph{Frege} proof system is one in which proof lines are 
(unquantified) propositional formulas derived using  rules given by schemes. 
This notion is very robust and any two
 such complete systems are polynomially equivalent. 
We will abuse notation slightly and take $G_0$ to be a Frege system,
even though the lines are formally sequents rather than formulas.

\emph{Substitution Frege} systems are Frege systems augmented with a 
substitution rule that allows one to derive a substitution instance from any formula. 
We will consider a sequent calculus version, which is $G_0$ plus the substitution rule
\bel{e-substitution}
\frac{\Gamma\longrightarrow\Delta}{\Gamma[b/\alpha]\longrightarrow\Delta[b/\alpha]},
\ee
where $b$ is a free variable and $\alpha$ is a quantifier-free formula. 
The polynomial simulation between Frege systems and the propositional sequent calculus readily extends to the systems with substitution rules. Again slightly abusing notation, we take substitution Frege to be this system and abbreviate it as~$\SF$
(in fact we will only need a restricted form of this rule, where we substitute into a single formula in the antecedent).
% and assume that the family $SF$ also contains this proof system.

\emph{Extended Frege}, EF, adds to Frege the ability to introduce, for any formula~$\phi$, a new propositional
variable symbol $e_\phi$ and the \emph{extension axiom} $e_\phi \longleftrightarrow  \phi$.
 This potentially allows us to substantially
shorten proofs. It is well-known that SF and EF polynomially simulate each 
other~\cite{dowd,krajicek-pudlak89}.
We will use EF as an auxiliary proof system but will not consider large, implicit EF proofs\footnote{
For this reason, we do not give full details about the formalization of EF.};
for example Theorem~\ref{the:main} is about $\nimpres(\SF)$ rather than $\nimpres(\EF)$.
This is even though EF is very common in the literature while SF is  rare.
The reason is that the ability to abbreviate formulas by extension variables 
means that in EF we only ever need narrow proof-lines, so
$\nimpres(\EF)$ is polynomially equivalent to $\impres(\EF)$~\cite{krajicek-implicit}
which is, most likely, a much stronger proof system than $\nimpres(\SF)$,
since it polynomially simulates the full system~$G$~\cite{krajicek-implicit}.

\subsection{Bounded arithmetic}

We will use the bounded arithmetic theories $\PV$, $S^1_2$ and $T^1_2, T^2_2, \dots$. We take these all 
to be in the language $L_\PV$, which has a symbol for every polynomial time machine defining a 
function or relation on binary strings. Where appropriate we identify strings with natural numbers.  
An important symbol is $|x|$ for (a standard machine defining) the length
of the string~$x$.
A $\Sigma^b_1$-formula
has the form $\exists{y {<} t(x)} \, \phi(x,y)$, where $\phi$ is a quantifier-free $L_\PV$-formula and~$t$ is a $L_\PV$ term.
We define $\Pi^b_1$ etc.  similarly.
The theory $\PV$ then consists of universal axioms fixing the standard properties of polynomial-time machines.
The theory $S^1_2$ adds to this the \emph{$\Sigma^b_1$ length-induction} scheme
\[
\phi(0) \wedge \forall i {<} |x| \, [\phi(i) \rightarrow \phi(i+1)] \
\longrightarrow \
\phi(|x|)
\]
where $\phi$ is $\Sigma^b_1$ and may contain other parameters.
The theory $T^1_2$ strengthens this further to the full $\Sigma^b_1$ induction scheme,
which has the same form except that the length-bound $|x|$ is replaced with simply $x$,
so that exponentially longer inductions are available. 
The theory $T^i_2$ generalizes this to the full induction scheme
for~$\Sigma^b_i$ formulas.
For details see
e.g.~\cite{buss-book, krajicek-book2}.

The theories $\PV$ and $S^1_2$ are very similar in strength and are closely connected
with the propositional system EF, in that a proof in either
 of a universal $L_\PV$ formula $\forall x \phi(x)$ should be thought of as
a very succinct description of a family of polynomial-size EF proofs
of the translation of $\forall x \phi(x)$ into a family of propositional tautologies. 
The theory $T^i_2$ has a similar relation with~$G_i$, and we will 
also use a connection between $T^1_2$ and the TFNP 
class~PLS~\cite{krajicek-pudlak90, buss-krajicek}.

We will use the translation of $S^1_2$ into EF as a tool for constructing implicit EF proofs.
A rule of thumb for proving things in $S^1_2$ 
is that any properties of a polynomial-time algorithm which can be proved only using
elementary combinatorics can be formalized in~$S^1_2$.

%%%%%%%%%%%%%%%%%%%%%%%%%%%%%%%%%%%%%%%%%%%%%%%%%%
%%%%%%%%%%%%%%%%%%%%%%%%%%%%%%%%%%%%%%%%%%%%%%%%%%
\section{Implicit and narrow implicit proofs} \label{sec:implicit}
%%%%%%%%%%%%%%%%%%%%%%%%%%%%%%%%%%%%%%%%%%%%%%%%%%
%%%%%%%%%%%%%%%%%%%%%%%%%%%%%%%%%%%%%%%%%%%%%%%%%%

The restricted version of $\imp(P,Q)$, in which lines in the implicit $Q$-proof are polynomial size,
 was defined by Kraj\'i\v cek in~\cite{krajicek-implicit} with the notation $[P,Q]^m$.
 We will denote it by $\nimp_P(Q)$ and call it \emph{the narrow implicit proof system defined by $P, Q$}. 
 
 Unlike the general concept of implicit proofs, narrow implicit proof systems are only defined when the proof system $Q$ is based on formulas. More precisely, we will say that $Q$ is a \emph{formula-based proof system}, if $Q$ is given by a class of syntactic objects $\cal F$, called \emph{formulas}, and a finite set of relations~$D$ on~$\cal F$, 
 called \emph{deduction rules}, such that a $Q$-proof of $\phi$ is a sequence $\Pi=(\phi_1\dts\phi_t)$
  such that for every $i$, $\phi_i$ is in $\cal F$ and follows from some $\phi_{i_1}\dts\phi_{i_k}$ with $i_1\dts i_k<i$ by a deduction rule (formally, $(\phi_{i_1}\dts\phi_{i_k},\phi_i)\in R$ for some $R\in D$). Furthermore, we assume that formulas are encoded as binary strings and the relations defining the rules are decidable in polynomial time (assuming this encoding of formulas).
%We will also need to encode explicitly the indices $i_1\dts i_k<i$ of formulas $\phi_{i_1}\dts\phi_{i_k}$ from which the i-th formula follows for every~$i$.

Formula-based proof systems may be viewed as a generalization of Frege systems. We leave it to the reader to define refutations and proofs from assumptions in these systems. 
In most of the proof systems used in this paper, the lines are strictly sequents rather than formulas.
In order to view them as formula based proof systems, 
we will treat sequents as formulas and deduction rules as being applied to sequents.
If we write a  sequent as a conjunction of formulas implying a disjunction of formulas, then the resulting calculus will be a generalization of a Frege calculus where the inferences depend also on subformulas. Such rules are called
\emph{deep inference rules}
 (although in our case  they only manipulate subformulas one or two levels down).

\bdf
Let $P$ be an arbitrary proof system and $Q$ a formula-based proof system with a class of formulas $\cal F$ and a set of deduction rules $R_1\dts R_l$. A proof of a formula $\phi$ in $\nimp_P(Q)$,
the narrow implicit proof system defined by $P, Q$, is a pair $(\Pi,C)$ where $C$ is a circuit computing a Boolean function $f_C:\{0,1\}^n\to\{0,1\}^m$ and $\Pi$ is a $P$-proof of the formula $\gamma_{C,\phi}$ that says that $f_C$ correctly encodes a $Q$-proof of~$\phi$. In more detail,~$\gamma_{C,\phi}$ should express the following condition,
where we view strings in~$\{0,1\}^n$ as numbers in the interval $[2^n]$:
\bi
\item 
For $i\in \{0,1\}^n$,  the bit-string $C(i)$ output by $C$  encodes a tuple 
$(\phi_i;i_1\dts i_k;j)$ of a formula, numbers $i_1\dts i_k<i$ and a number~$j$,
 such that $\phi_i$ follows from $\phi_{i_1}\dts\phi_{i_k}$ using rule~$R_j$
\item 
$\phi$ is the final string $C(2^n{-}1)$ (but see below).
  %expresses that  $\phi$ is true; or some similar condition.
\ei
%For formalizing $C$ in the propositional calculus, we use variables for every vertex of $C$ and clauses expressing that the values correspond to the gates at the vertices; for comparing numbers, we use the standard formula, for checking that $\phi_i$ follows from $\phi_{i_1}\dts\phi_{i_k}$ using rule~$R_j$, we use a specific circuit for each $j=1\dts l$ and represent it as a formula using extension variable (in the same way as for $C$).
\edf

This second condition is a template that needs specification in concrete cases of  a formula-based system,
keeping in mind that we are blurring the distinction between formulas and sequents.  
In the case of a $G$ proof of a quantifier-free formula~$\theta$,
the last line should strictly encode, as a formula, the sequent  $\longrightarrow \theta$.
In the case of resolution,
where we are refuting some CNF $D_1 \wedge \dots \wedge D_m$ and where lines are clauses, the condition should 
rather be that the first $m$ lines
are the clauses $D_1, \dots, D_m$ and the last line is the empty clause.

Let us say more explicitly how we formalize $\gamma_{C, \phi}$. 
Think of the algorithm that, on input~$i$, does the following.
First it checks the second condition by computing $C(2^n{-}1)$ (regardless of $i$)
 and halting with output~$0$ if this does not satisfy the condition.
Then, to check the first condition at~$i$, it computes~$C(i)$;
checks that it encodes a tuple $(\phi_i;i_1\dts i_k;j)$; 
checks~$i_1\dts i_k<i$; computes all strings $C(i_1), \dots , C(i_k)$;
recovers the formulas $\phi_{i_1}, \dots, \phi_{i_k}$; and finally outputs~1 
if~$R_j(\phi_{i_1}, \dots, \phi_{i_k}, \phi)$ holds and~0 otherwise.
  We can compute
 this algorithm by a circuit $D$ with $n$ inputs (for the bits of~$i$). 
 The formula $\gamma_{C, \phi}$ expresses that $D$ outputs~1 on all inputs.
 This is a tautology if~$C$ encodes a syntactically correct proof.
%
%In particular, it is necessary to define  $\gamma_{C,\phi}$ explicitly.
%The formula $\gamma_{C, \phi}$ expresses a universal property of the circuit~$C$ on tuples of
%inputs, and to define it we also  assume we have circuits expressing the rules $R_i$ and the property
%of being a formula in~$\cal F$. 

We would like to use resolution as our auxiliary proof system. This is formally a refutational system,
so we specify a CNF $\neg \gamma_{C, \phi}$ that we must refute 
rather than a formula $\gamma_{C, \phi}$ that we must prove.
The CNF $\neg \gamma_{C, \phi}$ has variables for
the bits of $i$ and for the computation of $D$ on this input, and expresses 
``there is a computation of $D$ which outputs $0$ on this input''.
Here we formalize a computation of a Boolean circuit using a CNF in the usual way. 
That is, we introduce a variable $x$ for every vertex $g$ of the circuit, to express the Boolean value computed at this vertex. 
If $g$ is an input, we introduce clauses saying that $x$ is equivalent to the appropriate input variable.
Otherwise, we introduce clauses that determine the value of~$x$ from the values~$y,z$  of its predecessor gates in~$D$. 
If the gate at $g$ is an~AND, then the clauses state that $x$ is equivalent to $y\wedge z$, and similarly for other gates.

\subsection{From extended Frege to resolution}

\emph{Extended resolution}, ER, is a proof system that adds to resolution the ability to introduce, for any
two literals $x, y$, a new variable $z$ and three clauses
$\neg x \vee \neg y \vee z$, $\neg z \vee x$ and $\neg z \vee y$ expressing together that $z \equiv x \wedge y$. 
It is easy to see that ER is polynomially equivalent to EF. 
On the other hand, these extension clauses are exactly the same as the clauses we just described that
force the value at an AND gate in a circuit to be computed correctly. We can use this observation
to show that resolution and extended resolution, and thus extended Frege, are equivalent as auxiliary
proof systems in implicit proofs.

The following lemma (there for $\impres$ rather than $\nimpres$)  is from~\cite{wang}.

\begin{lemma}  \label{lem:EF_to_Res}
Let $P$ be any proof system. Then $\nimpres(P) \equiv_p \nimp_\EF(P)$.
\end{lemma}

\begin{proof}
The hard direction is 
$\nimp_\EF(P) \le_p \nimpres(P)$.
So suppose we are given a $\nimp_\EF(P)$ proof of a formula~$\phi$,
that is a pair $(\Pi, C)$ where $\Pi$ is an $\EF$ proof of the formula $\gamma_{C, \phi}$
described above, expressing that $C$ encodes a $P$-proof of~$\phi$. 
We may assume $\neg \gamma_{C, \phi}$ is a CNF.

From $\Pi$ we construct in polynomial time an ER refutation $\Pi_\mathrm{ER}$ of $\neg \gamma_{C, \phi}$.
We may assume that $\Pi_\mathrm{ER}$ has the form of a sequence $\Delta$ of extension axioms 
defining extension variables $x_1, \dots, x_m$ from the variables of $\neg \gamma_{C, \phi}$, 
followed by a resolution refutation $\Pi_\mathrm{Res}$ of $\neg \gamma_{C, \phi} \cup \Delta$.
By construction each variable in $\neg \gamma_{C, \phi}$ 
represents the value computed at some gate in the circuit~$D$ which checks
correct behaviour of~$C$ on input~$i$. Therefore for each new extension variable~$x_j$
introduced in $\Delta$, we can add a gate to~$D$ which computes the value of~$x_j$ based on
the axioms in~$\Delta$. 
We call the resulting circuit~$D'$.

Let $C'$ be the circuit which, on input $i$, first runs the circuit~$D'$ on $i$, throws away
any output, and then runs the circuit~$C$ on~$i$ and outputs~$C(i)$.
Then any computation of $C'$ on~$i$ can be viewed as a computation of $C$ on~$i$,
if we ignore the unused preliminary run of~$D'$,
so that with some renaming of variables $\Pi_{\mathrm{Res}}$ is still 
a resolution refutation of $\neg \gamma_{C', \phi} \cup \Delta$
(where we have replaced $C$ with $C'$ in the subscript).
But now each extension variable introduced in $\Delta$ is associated with some
node~$u$ of $D'$. We can use variable at~$u$ in the computation of $C'$ as
the extension variable, and using the relation between the structure of $D'$ and the structure
of $D$ we can derive the corresponding extension clause. In this way we can construct
from $\Pi_{\mathrm{Res}}$ a resolution refutation $\Pi'$ of $\gamma_{C', \phi}$
which does not introduce any formal extension variables. Then $(\Pi', C')$
is the required $\nimpres(P)$ refutation of $\phi$.
\end{proof}

\subsection{$G_{i+1}$ simulates $\nimpres(SG_i)$}

Equipped with the definitions we can now prove one direction of our main theorem.
We prove Lemma~\ref{lem:G_i+1_simulates},
that $G_{i+1} \ge_p \nimpres(SG_i)$ for $i \ge 0$.

We use a (tight) connection between $G_{i+1}$ and the bounded arithmetic theory~$T^{i+1}_2$.
The soundness of a proof system $P$ is the statement that, given a triple $(\phi, \pi, \alpha)$, if $\pi$
is a $P$-proof of the formula $\phi$ and $\alpha$ is an assignment to the variables of $\phi$,
then $\alpha \vDash \phi$. It follows from Theorem~9.3.17 of~\cite{krajicek-book1} (originally proved in \cite{krajicek-pudlak90}) that, if~$T^{i+1}_2$ proves
the soundness of $P$, then $G_{i+1}$ simulates~$P$.

We describe a proof of soundness of $\nimpres(SG_i)$ which can be formalized in $T^{i+1}_2$.
Let $(\Pi,C)$ be a $\nimpres(SG_i)$ proof ending with the sequent~$\seq\phi$. 
For easier reference, we denote by $\Sigma$ the large proof defined by~$C$. 
Thus $\Pi$ is a resolution proof of the fact that the steps of $\Sigma$ (as defined by $C$) are sound. 
Since even $S^1_2$ proves the soundness of (explicit) resolution proofs, working in $T^{i+1}_2$
we can see that all steps in~$\Sigma$ are correct. 
We will now use $\Pi^b_{i+1}$ induction on~$m$, with the inductive hypothesis 
``each of the first~$m$ sequents of $\Sigma$
is satisfied by every assignment to its free variables''. 
Here it is important that the implicit proof is narrow, as we can write a $\Pi^b_{i+1}$ formula that 
evaluates a polynomial-size formula (here expressing a sequent) 
over all assignments, while we would not know how to do this for
an exponential-size formula.

For the induction to go through, we need to show that the axioms are satisfied by all assignments and that this property is preserved by the steps in proof~$\Sigma$. Consider the case when a sequent ${\Gamma[b/\alpha]\longrightarrow\Delta[b/\alpha]}$ is derived from ${\Gamma\longrightarrow\Delta}$ using the substitution rule, the only interesting case. One can easily see that
our theory can prove that the satisfiability by all assignments is preserved; indeed, given an assignment for the variables of  ${\Gamma[b/\alpha]\longrightarrow\Delta[b/\alpha]}$,
we only need to evaluate $\alpha$ and use it in the assignment for ${\Gamma\longrightarrow\Delta}$.

Thus in particular the last sequent $\seq\phi$ is a tautology. Hence the original proof  $(\Pi,C)$ is sound.
% The actual proof in $T^1_2$ will not talk directly about sequents in $\Sigma$, but about the indices of sequents in~$\Sigma$. So the induction will proceed over the indices.

\subsection{From bounded arithmetic to implicit proofs}

For the other direction of the main theorem, where we need to construct implicit proofs, we will use
the  following observation, essentially from~\cite{krajicek-implicit}.

\begin{lemma} \label{lem:s12-implicit}
Let $P$ and $Q$ be proof systems where $Q$ is formula-based. 
Suppose there is a $\PV$-term $f$ and a polynomial $p$ such
that, provably in $S^1_2$, if any string $\pi$ is a $P$-proof of a formula $\phi$
then the sequence of strings $f(\pi, 0), \dots, f(\pi, 2^{p(|\pi|)}{-}1)$ is a $Q$-proof of $\phi$. 
Then $P \le_p \nimp_{\EF}(Q)$.
\end{lemma}

\begin{proof}
The assumptions tell us that $S^1_2$ proves a certain $\forall \Pi^b_1$ sentence.
By standard results about $S^1_2$, it follows that propositional translations
of this sentence have polynomial-time constructible $\EF$ proofs. That is, there is a sequence
$\{ \pi_n \}_n$ of EF proofs of the tautology
(\ref{eq:EF-proof}) below, where $\pi_n$ is constructible in polynomial time
from the unary string $1^n$.
The tautology has Boolean variables $x_1, \dots, x_n$ coding  an $n$-bit $P$-proof; 
$y_1, \dots, y_n$ coding the proved formula $\phi$; and $i_1, \dots, i_{p(n)}$ coding a location $i$
in the implicit $Q$-proof computed by~$f$:
\begin{equation} \label{eq:EF-proof}
\ang{\textrm{$\vec x$ codes a $P$-proof of $\phi$}}
\rightarrow
\ang{\textrm{$f(\vec x, \cdot)$ locally at~$i$ is a $Q$-proof of $\phi$}}.
\end{equation}

Now we can show the reduction. Suppose we are given a concrete pair~$\phi, \rho$ of a formula $\phi$
and a $P$-proof $\rho$ of $\phi$. Then $\rho$ has some binary length~$n$. We generate the EF proof
$\pi_n$ from the sequence above, whose last line has the form~(\ref{eq:EF-proof}). We substitute
into~(\ref{eq:EF-proof}) the actual bits of $\rho$ for the variables $\vec x$, and the actual bits of $\phi$ 
for the variables $\vec y$ describing the formula~$\phi$ in~(\ref{eq:EF-proof}).
This will satisfy the left-hand expression in~(\ref{eq:EF-proof}), so
we obtain an EF proof of a formula
\begin{equation} \label{eq:EF-proof-2}
\ang{\textrm{$f(\rho, \cdot)$ locally at~$i$ is a $Q$-proof of $\phi$}}
\end{equation}
where $i$ is coded as variables $i_1, \dots, i_{p(n)}$ and $\rho$ and $\phi$ are now fixed.
We may assume that, in the propositional translation, $f(\rho, \cdot)$ is now represented
as a fixed circuit~$C$ taking as input the variables $i_1, \dots, i_{p(n)}$. 
Thus from~(\ref{eq:EF-proof-2}) we can straightforwardly derive
the formula $\gamma_{C, \phi}$. Thus we have a circuit~$C$
and an EF proof of~$\gamma_{C, \phi}$, which together comprise
a $\nimp_{\EF}$ proof of~$\phi$.
\end{proof}

\begin{corollary} \label{cor:s12_nimpres}
Under the assumptions of Lemma~\ref{lem:s12-implicit},
$P {\le_p} \nimpres(Q)$. \qed
\end{corollary}

%%%%%%%%%%%%%%%%%%%%%%%%%%%%%%%%%%%%%%%%%%%%%
%%%%%%%%%%%%%%%%%%%%%%%%%%%%%%%%%%%%%%%%%%%%%
\section{Unfolding DAGs into trees}\label{sec:unfolding}
%%%%%%%%%%%%%%%%%%%%%%%%%%%%%%%%%%%%%%%%%%%%%
%%%%%%%%%%%%%%%%%%%%%%%%%%%%%%%%%%%%%%%%%%%%%

Unfolding DAGs, directed acyclic graphs, is one of the basic constructions used in complexity theory
and will play an important role in our results. 
Trees in this paper are finite with indegree at most two and with edges directed towards the root, 
which we think of as at the bottom of the tree.
We will sometimes call edges \emph{arrows}.

Let $G$ be a rooted DAG, that is, a directed acyclic graph with only one sink,
which furthermore has indegree at most two and labeled vertices. 
Let~$V$ be the set of vertices of $G$ and $\lambda$ the labeling function. 
We will think of the predecessors of a node as being given in some particular
order, so we also associate with $G$ two functions $L$ and $R$
for the left and right predecessors.
On vertices of indegree one $L$ and $R$ coincide, and on leaves we take them both to be loops.

Then such  a $G$ can be unfolded into a labeled rooted tree $T$ of size at most~$2^{|V|-1}$,
equipped with functions $\lambda$, $L$ and $R$,
in such a way that the ``local structure'' is preserved, 
meaning that there is a homomorphism from~$T$ to~$G$.
In the presence of $L$ and $R$ this homomorphism is unique.

\subsection{Succinct representation of trees}\label{sss-trees}

We will need to define the structure of $T$ succinctly, using a circuit. 
To this end we need a concrete representation of~$T$. 
The basic idea is that we can address a vertex $v$ of $T$ by a path in $G$ going up from the root of $G$ 
to the homomorphic image of~$v$. If we use $0$ for the left predecessor and $1$ for the right predecessor, then each node gets an address consisting of a binary string of length at most $n-1$. However, we need to represent the addresses by strings of equal length  because the strings will be inputs of Boolean circuits, so we also need a blank symbol to pad out the string.
Moreover we would like to use an encoding on which the natural, lexicographic ordering on addresses has the following two properties, for nodes~$u,v,w$:
\bi
\item if $u$ is above $w$ in the tree, then $u < w$
\item if $u$ is in the right subtree above $w$, and $v$ is in the left
subtree above~$w$, then $u<v$.
\ei
Here the first condition is clearly desirable, and we want the second because of the particular
re-wiring we do of a proof during cut-elimination.

We can achieve this by using strings of length~$2n$ in which right and left are represented
by $00$ and $01$ respectively, and the blank symbol is represented by~$11$. Thus
the node reached from the root by going first left and then right
has the address $01 \, 00 \, 11 \, \dots \, 11$. 

Then the relation that defines the arrows of $T$, the predecessor functions, and the successor functions can be computed by polynomial-size circuits taking these addresses as inputs.
To compute the labels of $T$, we first observe that the homomorphism from~$T$ to~$G$ can be computed by a small circuit. Indeed, given $\vec v$ representing a vertex of $T$, the circuit can compute the path in $G$ that leads to the homomorphic image of $\vec v$ in $G$. Having determined the vertex $u$ in $G$, the circuit can use a lookup table to determine the label of $u$.

\subsection{Verifying properties}\label{s-verifying}

Next we need to show that some properties of~$G$ are preserved in $T$. 
Suppose~$G$ represents the %predecessor-successor relation 
edge-structure
of a proof in some formula-based proof system.
That is, the labels of $G$ are formulas and~$G$ satisfies the property that each formula is derived from its predecessors by a rule of the system (on leaves we have axioms, which we view as rules with no premises).
We need to be able to succinctly prove that unfolding this DAG into a tree produces a 
treelike proof in the same system.

Returning to our abstract labeled DAG $G$ and its unfolding~$T$,
let $h$ be the homomorphism from $T$ to $G$ defined above. 
We need to establish, by short uniform EF proofs or equivalently in~$S^1_2$, 
that $h$ satisfies three properties.
\ben
\item
For every $\vec v\in\{0,1\}^{2n}$, either $h(\vec v)$ is one of the vertices of~$G$ or,
in case~$v$ does not represent a vertex in $T$, then $h(\vec v)$ is a special symbol~$\diamond$ 
\item 
The function~$h$ is a homomorphism with respect to the predecessor functions, that is,
$L(h(\vec v))=h(L(\vec v))$
and 
$R(h(\vec v))=h(R(\vec v))$,
where $L(\vec v)$ and $R(\vec v)$ represent left and right predecessors in~$T$.
\item 
The function~$h$ preserves labels, that is, $\lambda(\vec v)=\lambda(h(\vec v))$.
\een
It is clear that if we define $h$ as in the previous section, then these facts have short proofs. 
For instance, we defined the labels on $T$ using the equation $\lambda(\vec v)=\lambda(h(\vec v))$, so we get the third property immediately.

Once we have these proofs, proving other things becomes easy. Suppose, for example, that there is a ternary relation $A$ defined on labels such that 
$
A\bigl(\, \lambda(u), \, \lambda(L(u)), \, \lambda(R(u)) \, \bigr)
$
holds in $G$ for every vertex~$u$, perhaps defining that formulas are derived according to a certain rule.
The proof that $T$ has the same property is as follows:
%\begin{quote}\small {\it Proof.}
 Let $\vec v\in\{0,1\}^{2n}$. If $h(\vec v)\neq\diamond$ then 
 \mbox{$h(\vec v)=u$} for some vertex in~$G$. So we have
$A\bigl(\, \lambda(h(\vec v)), \, \lambda(L(h(\vec v))), \, \lambda(R(h(\vec v))) \, \bigr)$.
Using~2. and 3. above, we get
\[
\lambda(h(\vec v))=\lambda(\vec v),\quad \lambda(L(h(\vec v)))=\lambda(L(\vec v)) \quad
\mbox{and} \quad \lambda(R(h(\vec v)))=\lambda(R(\vec v))
\]
and conclude
$A\bigl(\, \lambda(\vec v), \, \lambda(L(\vec v)), \, \lambda(R(\vec v)) \, \bigr)$.

\subsection{An example proposition}

We conclude this section with a simple proposition 
(which we will not use later)
that we can  derive using the considerations above.

\bpr
Let $P$ be a formula based proof system and $P^*$ its treelike version. Then $\nimpres(P^*)$ polynomially simulates~$P$.
\epr

\begin{proof}
We can use the unfolding technique described above
to construct a polynomial-time function $f(\pi, i)$ which, given a $P$-proof~$\pi$ of a formula $\phi$,
 will tell you any given line of the unfolding
of $\pi$ into a treelike proof. Furthermore we can prove in $S^1_2$ that the object described by $f$
is still a $P$-proof. The proposition follows by Corollary~\ref{cor:s12_nimpres}.
\end{proof}

%%%%%%%%%%%%%%%%%%%%%%%%%%%%%%%%%%%%%%%%%%%%%%%%%%
%%%%%%%%%%%%%%%%%%%%%%%%%%%%%%%%%%%%%%%%%%%%%%%%%%
\section{Uniform cut-elimination} \label{sec:main}
%%%%%%%%%%%%%%%%%%%%%%%%%%%%%%%%%%%%%%%%%%%%%%%%%%
%%%%%%%%%%%%%%%%%%%%%%%%%%%%%%%%%%%%%%%%%%%%%%%%%%

This section contains the main construction of the paper,
proving Theorem~\ref{the:cutelim_theorem} by
using cut-elimination to turn an explicit (small) $G_{i+1}$ proof
into an implicit (large)~$G_i$ proof. 
Precisely, in Sections~\ref{sec:aux_rules} and~\ref{sec:absorb} we 
introduce some auxiliary rules and proof systems; in Section~\ref{sec:imp_SG_0}
we show  that $G_1 \le_p \nimpres(S_L G_0)$, where~$S_L G_0$ is $G_0$
plus a special case of the substitution rule; and in Section~\ref{sec:cut_elim_general} we
generalize this to $G_{i+1} \le_p \nimpres(S_L G_i)$ and finally show
that for $i \ge 1$ we can eliminate substution to get $G_{i+1} \le_p \nimpres(G_i)$.

\subsection{Auxiliary rules} \label{sec:aux_rules}

We will make use of two auxiliary rules.
\bi
\item 
The \emph{substitution-left} rule is the following minor variant of the substitution rule:
\[
\frac{\Gamma, \phi(a) \seq \Delta}{\Gamma, \phi(a / \beta) \seq \Delta}
\]
where $\beta$ is a quantifier-free formula and $a$ does not occur in $\Gamma$ or $\Delta$.
\item 
The $S_LWC$ rule (short for \emph{substitution-left, weakening, cut}) is
\bel{e-s-c}
\frac{\Gamma' \seq \phi(a_1 / \beta_1\dts a_k / \beta_k),\Delta' \quad\quad
\phi(a_1\dts a_k),\Gamma'' \seq\Delta''}{\Gamma\seq\Delta}
\ee
where $\beta_1\dts \beta_k$ are quantifier-free formulas, 
$a_1\dts a_k$ are free variables which do not occur in $\Gamma$ or $\Delta$, 
and $\Gamma', \Gamma'' \sub\Gamma$ and $\Delta', \Delta'' \sub \Delta$. 
\ei
An instance of the  $S_LWC$ rule can be replaced with~$2k$ instances of substitution-left,
replacing $\phi(a_1\dts a_k)$ in the right premise by $\phi(a_1 / \beta_1\dts a_k / \beta_k)$, 
followed by two weakenings, followed by a cut
(this requires $2k$ substitution-lefts, rather than just~$k$, because an $a$~variable
may occur inside a $\beta$ formula; so we first need to substitute-in all the $\beta$ formulas
with any such variables replaced with dummies, then substitute-in again to replace
the dummy variables with $a$-variables).
 We include it as a single rule only so that 
the surgery on implicit proofs during cut elimination works more tidily.

We denote by $S_{L}G_i$ the extension of $G_i$ in which we also allow the
substitution-left rule.

\subsection{Absorption systems} \label{sec:absorb}

As an intermediate step in our main theorem,
 it will be helpful to put $G_i$ proofs into a highly structured form that uses specialized rules.
We start by defining this.

By an \emph{absorption} system we broadly mean a sequent calculus with two properties. 
Firstly, all formulas in the conclusion of a rule must also be present 
in every premise. One consequence of this is that every formula in the proof must already appear in some axiom.
Secondly, they require explicit parent-child relations between formulas in the premise and formulas
in the conclusion, such that
\bi
\item
Each formula in the conclusion has, in each premise, an occurrence of the same formula
as a parent, and may also have one or more auxiliary formulas as parents; and
\item
No formula in a premise has more than one child.
\ei

Given a rule of $G$ (except weakening), or the substitution-left rule, we form its \emph{absorption version} as follows:
\ben
\item 
Mark each side formula in a premise as a parent of the corresponding side formula
in the conclusion
\item
Mark the auxiliary formulas as parents of the principal formula (if there is one)
\item
Finally add a new copy of the principal formula (if there is one) to every premise, and 
mark these as parents of the principal formula.
\een
We will call the resulting rules \emph{absorption $\exists$-left} etc.
For example the absorption $\vee$-left rule is
\[
\frac{A, A\vee B, \Gamma \seq \Delta \qquad B, A \vee B, \Gamma \seq \Delta}
{A \vee B, \Gamma \seq \Delta}
\]
where $A \vee B$ in the conclusion has four parents, namely  $A$ and $A \vee B$
on the left and $B$ and $A \vee B$ on the right, and each occurrence of a formula
in $\Gamma$ or $\Delta$ in the conclusion
has two parents, namely the same occurrence of the formula in $\Gamma$ or $\Delta$
in the left and right premises.

The full weakening rule does not really make sense in an absorption system so we will not include it;
however we do allow absorption repetition. Absorption contraction, as we have defined it, now has
the strange form that we combine \emph{three} copies of a formula into one (because of the extra copy
of the principal formula), but this is not a problem for us.

We denote by $AG_i$ the version of $G_i$ which only uses absorption rules. 
It is not hard to see that $AG_i$ is polynomially equivalent to $G_i$, but it is unlikely
that this is true in general for implicit (or even narrow implicit) $AG_i$ and $G_i$.

\subsection{Main construction} \label{sec:imp_SG_0}

\bt \label{the:base_case}
$\nimp_{Res}(S_L G_0)$ polynomially simulates $G_1$.
\et

\begin{proof}
We will use Lemma~\ref{lem:s12-implicit}. To this end we need to show that there is a polynomial time function $f$ such that, 
given a string encoding a $G_1$ proof~$\pi$ with final sequent $\seq \phi$, 
the function $f$ computes the lines of a (potentially exponentially
larger) $S_L G_0$ proof with the same final sequent.
We will do this by transforming $\pi$ in a sequence of steps, working in~$S^1_2$. The object 
constructed at each step will be a proof (or almost-proof) definable
by a polynomial-time algorithm. We need furthermore that this can be formalized in~$S^1_2$. We will not describe a formalization since the steps of the algorithms are simple operations with strings of symbols, which is what $S^1_2$ can easily formalize.

Steps 1 to 4 will turn $\pi$ into an implicit treelike $AG_1$ proof.
Steps 5 and~6 will carry out cut-elimination, removing 
cuts on quantified formulas and resulting in an implicit
daglike proof. These  steps introduce some rules
that break monotony of sequents, and the constructed proof  no longer
uses only absorption rules. The final step 7 does some cleaning up
to produce an implicit $S_L G_0$ proof.

We say that a proof has \emph{ordered cuts} if, for every instance of the cut
rule on a formula $\theta$ in $\Sigma^q_1 \setminus \Sigma^q_0$,
the cut formula appears in the succedent in the left-hand premise
and in the antecedent in the right-hand premise.

\paragraph{Step 1} We replace $\pi$ with a $G_1$ proof $\pi_1$ in 
which all axioms have the form $A {\seq} A$ with $A$ quantifier-free.
For example if $\exists x \theta(x) {\seq} \exists x \theta(x)$
was an axiom in $\pi$,
with $\theta(x)$ quantifier-free, we instead derive this sequent from 
the axiom $\theta(a) \seq \theta(a)$
by using the $\exists$-right rule then the $\exists$-left rule.

This requires at most a quadratic increase in proof size,
which may occur if the blocks of quantifiers are very long.
We can still code $\pi_1$ directly as a string rather
than implicitly. % as a circuit. 
Let $h$ be the height of~$\pi_1$.

\paragraph{Step 2}
As described in Section~\ref{sec:unfolding}, we unfold $\pi_1$ into a treelike $G_1$ proof~$\pi_2$
with ordered cuts,
coded implicitly. % as a circuit.
Note that~$\pi_2$ has the same height~$h$ as~$\pi_1$.

\paragraph{Step 3} We put the proof into  free-variable normal form.
Let us say that a variable is \emph{eliminated} at a sequent~$\sigma$
if it occurs in $\sigma$ but not in the conclusion of the rule
in which $\sigma$ is a premise (recalling that we are dealing with a treelike proof), or
if $\sigma$ is the final sequent. Free-variable normal form  means here
that each free variable is eliminated at exactly one sequent and only
occurs above that sequent.

To put a treelike proof into free-variable normal form is a matter of renaming variables.
However we need to do this in such a way that the resulting proof is still
polynomial-time. % coded by a small circuit. 
For this it is sufficient to give a polynomial-time procedure which, 
% given the circuit encoding $\pi_2$ and 
given the address~$s$ of a sequent in~$\pi_2$ and making calls to the machine encoding~$\pi_2$, returns the 
sequent at~$s$ in the proof after the renaming. 

Our procedure is as follows.
Let $\sigma$ be the sequent at~$s$ in $\pi_2$.
List all free variables $a_{i_1}, \dots, a_{i_k}$ occurring in $\sigma$. For each such variable $a_{i_j}$,
search down
the branch in $\pi_2$ until we reach the first sequent $\tau$ where this variable no
longer occurs (that is, the sequent after it is eliminated).  We rename $a_{i_j}$ as $a_{i_j,t}$
where $t$ is the address of $\tau$ in the tree,
or $t$ is a symbol~$*$ if the variable survives to the final sequent. 
This procedure is polynomial-time since there are at most~$h$ sequents to consider
between $s$ and the final sequent and each sequent has polynomial size.

We call the object after the renaming~$\pi_3$ and claim
that it is still a $G_1$ proof.
This is because firstly, any free variable which occurred in two sequents in a rule will have
its name changed identically in both sequents.  Secondly,
the renaming does not break any eigenvariable condition,
since two variables which were distinct in~$\pi_2$ are still distinct in~$\pi_3$.

\paragraph{Step 4} We transform $\pi_3$ into an $AG_1$ proof $\pi_4$.
Consider any sequent $\sigma$ in $\pi_3$ which is not the last sequent.
It is the premise of some rule, whose conclusion is a sequent~$\tau$. 
Mark each side formula in $\sigma$ as a parent
of the corresponding side formula in~$\tau$.
Mark all auxiliary formulas in~$\sigma$ as parents of
the principal formula in~$\tau$, if there is one.

Then work down the path which leads from~$\sigma$ all the way to the last sequent. Each 
time a new formula is introduced on this path below~$\sigma$, 
either as a principal formula or in an instance of weakening,
add a copy of it at the end of the succedent or antecedent of~$\sigma$,
to match the side on which it is introduced. This formula either occurs in~$\tau$ 
already because it is introduced there, or will occur in $\tau$ after processing because
it is introduced below~$\tau$; in either case mark the new occurrence in $\sigma$
as a parent of the occurrence in~$\tau$. 

Again the path from $\sigma$ to the root has length at most~$h$ 
so we have added at most polynomially many formulas to~$\sigma$
(it may be more than~$h$, if several formulas are introduced in a weakening step)
and the new proof is still narrow in our technical sense. Furthermore
we can construct the new~$\sigma$ in polynomial time, so 
we have a polynomial-time machine encoding~$\pi_4$. % small circuit for~$\pi_4$.

To see that $\pi_4$ is a valid $AG_1$ proof we observe that the leaves
are still valid axioms (now with many extra side formulas)
and that the renaming that we did in step 3 guarantees that
the eigenvariable condition still holds at each $\exists$-left rule.
Furthermore all formulas remain in weakly prenex form.

\bigskip

Our goal now is  to remove every cut on a $\Sigma^q_1 \setminus \Sigma^q_0$ formula,
that is, on a quantified formula.
Since $\pi_4$ still has ordered cuts, these have the form
\bel{e-cut2}
\frac{\Gamma\seq\exists x_1\dots\exists x_k\phi(x_1\dts x_k),\Delta\quad\quad
\exists x_1\dots\exists x_k\phi(x_1\dts x_k),\Gamma\seq\Delta}{\Gamma\seq\Delta}
\ee
where $\phi$ is quantifier-free,
with the cut formula in the succedent in the left premise and the antecedent in the right
premise, as shown.

We record an easy claim about~$\pi_4$ which we will need later in step 6.

\begin{claim} \label{cla:descendant_substitution}
Let $\phi(a / \beta)$ be an occurrence in $\pi_4$ of
a quantifier-free auxiliary formula in an instance of absorption $\exists$-right,
so that the rule replaces  $\phi(a / \beta)$ with $\exists x \phi(x)$. 
Then every descendant 
of $\phi(a/\beta)$ is a quantified formula of the form
$\exists x_1 \dots \exists x_k \phi'(x_1, \dots, x_k)$ with the property
that $\phi(a / \beta)$ is a substitution instance of~$\phi'(x_1, \dots, x_k)$,
treating $x_1, \dots, x_k$ here as free variables which in the substitution we replace
with quantifier-free formulas. 
\end{claim}

\begin{proof}
This is clearly true for the immediate child $\exists x \phi(x)$ of $\phi( a / \beta)$,
since $\phi(a / \beta)$ is obtained from $\phi(x)$ by substituting $\beta$ for $x$.
Then consider the sequence of descendants of $\phi(a / \beta)$. Since we 
are in a treelike, absorption proof, this is a well-defined sequence of formulas
of length at most~$h$. Since the first formula in the sequence is quantified
and all formulas are weakly prenex, by our choice of rules all formulas in the sequence
are quantified and the only way a child in the sequence can be different from
its parent is through an instance of the  absorption 
$\exists$-right rule
\[
\frac{v: \Gamma \seq \Delta, \  \exists x_1, \dots, x_k \phi''( b/\gamma, \vec{x}),  \
\exists y \exists x_1, \dots, x_k \phi''(y, \vec{x})}
{w: \Gamma \seq \Delta, \ \exists y \exists x_1, \dots, x_k \phi''(y, \vec{x})}
\]
where $\exists x_1, \dots, x_k \phi''( b/\gamma, \vec{x})$ is the parent
and we inductively assume that $\phi(a / \beta)$ is a  substitution instance 
of~$\phi''( b/\gamma, \vec{x})$. Then, just as in the base case,
$\phi(a / \beta)$ is also a  substitution instance 
of~$\phi''( y, \vec{x})$
\end{proof}

\paragraph{Step 5} Let $\theta = \exists x_1 \dots \exists x_k \phi(x_1, \dots, x_k)$ 
be any occurrence of a
quantified formula appearing in an  \emph{antecedent} in $\pi_4$
-- in particular  the right-hand premise of a cut rule on a quantified formula
contains such a formula, as in~(\ref{e-cut2}). Search down
the branch and let $s$ be the address of the last sequent in which a descendant
of~$\theta$ appears,  recalling that the final sequent is quantifier-free.
Replace $\theta$ with $\phi(b_{x_1, s}, \dots, b_{x_k, s})$ where
$b_{x_1, s}, \dots, b_{x_k, s}$ are new \emph{free} variables with names as shown.

Since $\pi_4$ is in weakly prenex form, and since axioms allow only
quantifier-free principal formulas, the only rules that could have
been invalidated by this change are $\exists$-left and cut.
A typical form of the absorption $\exists$-left rule in $\pi_4$ is
\[
\frac{\Gamma, \exists x_2\phi(a, x_2), \exists x_1 \exists x_2 \phi(x_1,x_2) \seq \Delta}
{\Gamma, \exists x_1 \exists x_2 \phi(x_1,x_2) \seq \Delta}
\]
where the two displayed quantified formulas in the premise are 
the parents of the displayed quantified formula in the conclusion 
(and where for simplicity we assume there is a single pre-existing bound
variable $x_2$; in general there could be several, or none). The eigenvariable~$a$
does not appear in the conclusion.

Because of the parent-child relationships, for the three displayed
quantified formulas, when we search down to find the last sequent in which they
appear, we will find the same sequent~$s$ in all three cases.
Therefore after our transformation this rule turns into
\[
\frac{\Gamma, \phi(a, b_{x_2,s}), \phi(b_{x_1,s},b_{x_2,s}) \seq \Delta}
{\Gamma, \phi(b_{x_1,s},b_{x_2,s}) \seq \Delta}.
\]
This is now a valid instance of the absorption substitution-left rule, 
replacing~$a$ with $b_{x_1, s}$, and we relabel it as such.

We call the result of this step $\pi_5$.
Note that it no longer contains any quantified formula in any antecedent.
It is not a valid proof 
in any of our systems, since in cuts of type (\ref{e-cut2})
we have changed the cut formula on the right so that it no longer matches the one on the left.
We will handle this in the next step. But we have shown that, except
for at these cut rules, it is a valid proof in $AG_1$ plus absorption substitution-left.

\paragraph{Step 6}
To construct $\pi_6$ we first delete every quantified formula $\theta$ in every \emph{succedent} 
in~$\pi_5$, so that the proof no longer contains any quantified formulas.
We claim that with some surgery to the underlying graph of
the proof, and some relabelling of rules, this can be made into
a valid proof in the system $AG_0$ plus
the absorption substition-left rule and the $S_LWC$ rule.

Such formulas $\theta$ are introduced in $\pi_5$ by a $\exists$-right rule or as a 
side formula in an  axiom,
and are used in either a cut or another $\exists$-right rule.
So we need to show how such rules can be made valid after deleting 
all such formulas~$\theta$.

In an axiom, removing a side formula leaves a valid axiom, so there is nothing to show.
The relevant cuts were of type (\ref{e-cut2}) in $\pi_4$. In $\pi_5$ we transformed
the cut formula in the left premise; however now we will no longer use that premise.
This is because we have now deleted the cut formula
in the right premise. Therefore the right premise is now identical to the conclusion,
and so we can relabel the rule as an instance of repetition, deriving
the conclusion from the original right premise.

This leaves the absorption $\exists$-right rule, which 
has the form
\[
\frac{v: \Gamma \seq \Delta, \phi(a/\beta), \exists x_k \phi(x_k)}
{w: \Gamma \seq \Delta, \exists x_k  \phi(x_k)}
\]
where we have marked the 
address $v$ and $w$ of these sequents in the tree.
If $\phi(a / \beta)$ is quantified, then in~$\pi_6$ we have deleted it
and both occurrences of $\exists x_k  \phi(x_k)$ from the proof, so we can relabel
this as a repetition rule. The remaining case is when $\phi(a / \beta)$ is quantifier-free.
To fix $\pi_6$ we have to show how to derive 
$\Gamma \seq \Delta$ from $\Gamma \seq \Delta, \phi(a/\beta)$.
We will use the $S_LWC$ rule.

We analyze the proof $\pi_4$, before we did steps~5 or~6.
Let $s$ be the address of the sequent in $\pi_4$ containing the last descendant of the 
occurrence~$\phi(a/\beta)$. By Claim~\ref{cla:descendant_substitution}, 
this descendant is a quantified formula of the form
$\exists x_1 \dots \exists x_k \phi'(x_1, \dots, x_k)$ with the property
that $\phi(a / \beta)$ is a substitution instance of~$\phi'(x_1, \dots, x_k)$. 
Furthermore this descendant must necessarily be a 
cut formula in the left premise of a cut rule. Thus we had
the following picture inside $\pi_4$:
\begin{prooftree}
\AxiomC{$\vdots$}
\noLine
\UnaryInfC{$v:\Gamma \seq \Delta, \phi(a/\beta), \exists x_k \phi(x_k)$}
\UnaryInfC{$w:\Gamma \seq \Delta, \exists x_k \phi(x_k)$}
\noLine
\UnaryInfC{$\vdots$}
\noLine
\UnaryInfC{$s: \Gamma' \! \seq \Delta', \exists x_1 \dots \exists x_k \phi'(x_1, \dots, x_k)$}
\AxiomC{$\vdots$}
\noLine
\UnaryInfC{$t: \Gamma', \exists x_1 \dots \exists x_k \phi'(x_1, \dots, x_k) \seq \Delta'$}
\BinaryInfC{$u:\Gamma' \seq \Delta'$}
\end{prooftree}
In $\pi_5$ we replaced the cut formula in $t$
with $\phi'(b_{x_1, t}, \dots, b_{x_k,t})$ and in the current step we have deleted
all remaining quantified formulas.
So in the ``proof'' we are dealing with currently, the picture above has become
\begin{prooftree}
\AxiomC{$\vdots$}
\noLine
\UnaryInfC{$v:\Gamma \seq \Delta, \phi(a/\beta)$}
\UnaryInfC{$w:\Gamma \seq \Delta$}
\noLine
\UnaryInfC{$\vdots$}
\noLine
\UnaryInfC{$s: \Gamma' \! \seq \Delta'$}
\UnaryInfC{$u:\Gamma' \! \seq \Delta'$}
\AxiomC{$\vdots$}
\noLine
\UnaryInfC{$t: \Gamma', \phi'(b_{x_1, t}, \dots, b_{x_k,t}) \seq \Delta'$}
\noLine
\UnaryInfC{\phantom{$\Gamma$}}
\noLine
\BinaryInfC{}
\end{prooftree}
(for clarity we are ignoring that we may also have deleted quantified formulas in $\Delta$ and $\Delta'$).

As noted above, $\phi(a/\beta)$ is a substitution instance of $\phi'(b_{x_1, t}, \dots, b_{x_k,t})$
formed by replacing $b_{x_1, t}, \dots, b_{x_k,t}$ with quantifier-free formulas.
Also, by the presence of the suffix~$t$, no variable from $b_{x_1, t}, \dots, b_{x_k,t}$
can occur in the proof below address~$t$, so they cannot occur in $\Gamma'$ or $\Delta'$.
Lastly, by the absorption condition applied inductively to the sequents from $w$ to $s$, we
know $\Gamma' \subseteq \Gamma$ and $\Delta' \subseteq \Delta$. 
Therefore
\begin{prooftree}
\AxiomC{$v:\Gamma \seq \Delta, \phi(a/\beta)$}
\AxiomC{$t:\Gamma', \phi'(b_{x_1, t}, \dots, b_{x_k,t}) \seq \Delta'$}
\BinaryInfC{$w:\Gamma \seq \Delta$}
\end{prooftree}
is a valid instance of the $S_LWC$ rule.

Therefore we add an edge in our proof tree from $t$ to $w$, and mark in the proof
that $w$ is derived from $v$ and $t$ by the $S_LWC$ rule. We do this for all instances
of absorption $\exists$-right of this form.
As desired this is now a valid proof in the system $AG_0$ plus
the absorption substition-left rule and the $S_LWC$ rule.
In particular the edge from $t$ to $w$ is valid because we are using an ordering on nodes
in the tree in which everything in the right subtree above $u$ comes before everything 
in the left subtree above~$u$, so $t<w$; see the definition of the ordering
in Section~\ref{s-verifying}.

\paragraph{Step 7.}
We make $\pi_6$ into a $S_L G_0$ proof $\pi_7$. To do this we first choose a suitable~$k$,
greater than twice the number of variables in any individual formula, and uniformly
pad out each step in $\pi_6$ by following it with~$k$ steps applying the repetition rule. 
Using the space this creates after each instance of the $S_LWC$ rule, we replace that rule
with substitution-left steps, weakenings and a cut. We replace every absorption version of a rule
with the non-absorption version of that rule, followed by a contraction step to remove the extra
copy of the principal formula. Finally we give each variable appearing in the final sequent
back its original name.

\bigskip

This completes the simulation of $G_1$ by implicit $S_L G_0$.
\end{proof}

\subsection{Generalization to $G_i$ and eliminating substitution} \label{sec:cut_elim_general}

We complete the proof of Theorem~\ref{the:cutelim_theorem}, that is,
that
$G_{i+1} \le_p \nimpres(G_i)$ for  $i \ge 1$,
and 
$G_1 \le_p \nimpres(SG_0)$. 

We showed in Theorem~\ref{the:base_case}
in the last section that $G_1 \le_p \nimpres(S_L G_0)$ which clearly gives the second statement.
We claim that, with minimal changes, the proof of  Theorem~\ref{the:base_case}
also shows that $G_{i+1} \le_p \nimpres(S_L G_i)$.
Namely we use all the same steps, but rather than eliminating quantified formulas,
we eliminate those that are not $\Pi^q_{i}$ or $\Sigma^q_i$ 
--- the true reason for starting with the special case was that it is easier to talk about quantified and quantifier-free formulas than about $\Sigma^q_{i+1}\setminus(\Pi^q_{i}\cup\Sigma^q_i)$ formulas and $\Pi^q_{i}\cup\Sigma^q_i$ formulas.
However in rules where we substitute a variable with a formula, the substituted formula must still
be quantifier-free.

Finally we want to remove instances of the substitution-left rule, that is,
\[
\frac{\Gamma, \phi(a) \seq \Delta}{\Gamma, \phi(a / \beta) \seq \Delta}
\]
where $\beta$ is quantifier-free and $a$ does not occur in $\Gamma$ or $\Delta$.
We use a  simplified case of the construction in~\cite{krajicek-pudlak90} 
for full substitution. 
We describe below the case where~$\beta$ does not contain~$a$. If $\beta$
does contain~$a$, we choose any variable~$b$ which does not occur in $\Gamma$, 
$\Delta$ or $\beta$ and do the construction below twice, first to replace $a$ with $b$ and
then to replace $b$ with~$\beta$.

Using induction on the complexity of $\phi$ and $\beta$ we derive from the axioms the two sequents
\[
A: \ a \equiv \beta, \phi(a / \beta) \seq \phi(a)
\quad \textrm{and} \quad
B: \ \seq \exists x ( x \equiv \beta).
\]
Then we replace the instance of the rule with the following steps.
\begin{enumerate}
\item
$\Gamma, \phi(a) \seq \Delta$ \quad (the original premise)
\item
$\Gamma, a \equiv \beta, \phi(a / \beta) \seq \Delta$ \quad (by cutting with $A$)
\item
$\Gamma, \exists x (x \equiv \beta) , \phi(a / \beta) \seq \Delta$  \quad  (by $\exists$-left,
noting that~$a$ does not occur in $\phi(a / \beta)$, $\Gamma$ or $\Delta$)
\item
$\Gamma, \phi(a / \beta) \seq \Delta$  \quad (by cutting with $B$).
\end{enumerate}

This whole derivation has size polynomial in the size of the original instance of the rule
and is constructible in polynomial time. To replace all instances simultaneously in an implicit proof,
we first compute an upper bound $d$ on the number of steps required for this derivation, over all instances.
Then to give space for these derivations, we pad out the implicit proof by replacing every
sequent in the proof with a sequence of $d$ copies of the same sequent, derived from itself $d-1$
times using the repetition rule. Then in each case of the $S_L$ rule we replace this dummy
sequence with the derivation above. 

\section{Implicit resolution} \label{sec:impres}
%%%%%%%%%%%%%%%%%%%%%%%%%%%%%%%%%%%%%%%%%%%%%%%%%%%%%%%
%%%%%%%%%%%%%%%%%%%%%%%%%%%%%%%%%%%%%%%%%%%%%%%%%%%%%%%

We observed in Section~\ref{sec:intro} that $\nimpres(\Res) \equiv_p \impres(\Res)$,
since the formulas in a resolution refutation are clauses, which cannot be bigger than the number
of variables in the CNF being refuted. 
Let us remark here that while this is true for our formalization of resolution, 
the strength of $\impres(\Res)$ is potentially sensitive to the syntactical choices made.
For example, 
to show $\nimpres(\Res) \le G_1$ we showed in the proof of Lemma~\ref{lem:G_i+1_simulates}
that $T^1_2$ proves the soundness of $\nimpres(\Res)$. For this to go through for 
$\impres(\Res)$ rather than $\nimpres(\Res)$, we would like to be able
to evaluate a proof-line in polynomial time (in the length of the circuit describing the proof).
If a line is a clause, written out as a sequence of literals from the original formula with no repetitions, then this is
straightforward. But if lines are allowed to be padded out, 
for example by repeating literals, introducing new variables by weakening, 
or including repeated symbols such as $\bot$ (perhaps occurring as a result of a restriction by a partial assignment), 
then this may no longer be true.

We now prove Theorem~\ref{the:impres}, that
$G_1 \le_p \nimpres(\Res)$.

A PLS problem~\cite{pls}, with a parameter~$w$, is a triple $(D_w, N_w, t_w)$
where $D_w$ and $N_w$ are polynomial time machines defining a predicate and a function
(using the parameter~$w$) and $t_w$ is a  quasipolynomial in~$w$.
We may assume that $D_w \subseteq [t_w+1]$, that $t_w \in D_w$ 
and that $N_w: [t_w+1] \rightarrow [t_w+1]$.
A \emph{solution} to the problem is any $y \in D_w$ such that either $N_w(y) \le y$
or $N_w(y) \notin D_w$.

Thus a PLS problem with no solution, if such a thing existed, would give rise to a finite DAG with no sink,
with nodes given by~$D_w$ and edges given by~$N_w$. Using such a  DAG as the underlying graph,
inverting the direction of the edges, it is easy to construct 
a resolution proof of the empty clause using just the repetition rule. This is the main idea used below.

\bigskip

\noindent{\emph{Proof of Theorem~\ref{the:impres}.}}
The soundness of $G_1$, as a system for refuting CNFs,
is a $\forall\Pi^b_1$  sentence which is provable in~$T^1_2$~\cite{krajicek-pudlak90}. 
By results in \cite{buss-krajicek, beckmann_buss_09} we know that if a $\forall \Sigma^b_1$
sentence is provable in~$T^1_2$, then the search problem of witnessing this sentence
is reducible to an instance of PLS, provably over $S^1_2$. 
Since the soundness of $G_1$ is $\forall\Pi^b_1$ in particular it is $\forall \Sigma^b_1$, 
so it gives us a ``degenerate'' case of this witnessing theorem, where we do not 
exactly obtain an instance of PLS that we can reduce our sentence to (since our sentence
does not really represent a search problem), but rather an instance
of PLS which has no solution assuming our sentence is false.

Specifically, from the witnessing theorem and the provability of the soundness of $G_1$
in $T^1_2$, we obtain a PLS problem $Q_w = (D_w, N_w, t_w)$ with parameter~$w$
such that, provably in $S^1_2$, \emph{if} $w$ is a triple $(\phi, \pi, \alpha)$ witnessing that 
the soundness of $G_1$ is false, that is, such that
\begin{itemize}
\item
$\phi$ is a CNF in variables $x_1, \dots, x_n$
\item
$\pi$ is a $G_1$ refutation of $\phi$
\item
$\alpha$ is an assignment satisfying $\phi$,
\end{itemize}
\emph{then} $Q_w$ has no solution.
Of course in the real world no such $w$ exists and all PLS problems always
have solutions; the point is that these facts are (presumably) not provable in $S^1_2$, while
the above  implication \emph{is} provable.\footnote{
PLS, $T^1_2$ and reflection for $G_1$ play a similar role in~\cite{kol-thapen-24}.}

We claim that we can construct in polynomial time from $\phi$ and $\pi$ a
$\nimpres(\Res)$ refutation of $\phi$, that is, a circuit describing a large resolution
refutation of $\phi$ together with a resolution proof that the circuit works. 
By Corollary~\ref{cor:s12_nimpres}, it is enough to construct the large refutation
using polynomial time machinery, in a way that can be formalized in $S^1_2$.

So working in $S^1_2$, we fix $\phi$ and $\pi$ in the triple $(\phi, \pi, \alpha)$ 
and for clarity of notation consider the PLS problem $Q_w$ described above,
with $w = (\phi, \pi, \alpha)$,
as being parametrized simply by~$\alpha$, so we may write it as 
 $Q_\alpha = (D_\alpha, N_\alpha, t_\alpha)$.
Without loss of generality we may assume that there is $k \in \N$
such that, for each $\alpha$, the problem $Q_\alpha$ has domain $D_\alpha \subseteq [2^{n^k}]$,
that $t_\alpha = 2^{n^k}{-}1$ and is in~$D_\alpha$, and that a solution
is precisely any $y \in D_\alpha$ with either $N_\alpha(y) \ge y$ or $N_\alpha(y) \notin D_\alpha$.

We know (because it is provable in $S^1_2$) that if a total assignment $\alpha$ satisfies~$\phi$, then
the problem $Q_\alpha$ has no solution. 
Therefore the following is true: for any  $\alpha$ satisfying $\phi$,
we have that $2^{n^k}{-}1 \in D_\alpha$ and
for every $y \in D_\alpha$ both $N_\alpha(y) < y$ and $N_\alpha(y) \in D_\alpha$.

We now  describe our implicit refutation of $\phi$, still working inside~$S^1_2$.
For any total $\alpha$, let $C_\alpha$ be the clause
expressing the negation of $\alpha$. The refutation will be in two parts.
The ``upper part'' of the refutation will consist of, for each $\alpha$,
a resolution derivation $\Pi_\alpha$ of $C_\alpha$ from $\phi$.
All the refutations $\Pi_\alpha$ will be disjoint.
For the ``lower part'' of the refutation we can thus use a binary tree of height~$n$ to
resolve all the clauses~$C_\alpha$ together to derive the empty clause.

The derivation $\Pi_\alpha$ can take two different forms, depending on 
whether or not $\alpha$ falsifies $\phi$. In the standard model $\alpha$
always falsifies $\phi$, since $\phi$ is a contradiction, but we are working in $S^1_2$
so cannot assume this. What we will do instead is describe some polynomial
time machinery which uses $\alpha$ as a parameter and defines an object which, 
in any model of $S^1_2$, is a resolution derivation
with the required property.

\emph{Case 1.}
$\alpha$ falsifies $\phi$. Then we let $\Pi_\alpha$ derive $C_\alpha$ by a single weakening step 
from the first clause of $\phi$ which is false in $\alpha$.

\emph{Case 2.}
$\alpha$ satisfies $\phi$. Then in $S^1_2$, as we observed above, the PLS instance~$Q_\alpha$ has no solution
and in particular we have a polynomial-time set $D_\alpha \subseteq [2^{n^k}]$
and a polynomial-time function $N_\alpha$, such that $2^{n^k}{-}1 \in D_\alpha$ and 
for every $y \in D_\alpha$ both $N_\alpha(y) < y$ and $N_\alpha(y) \in D_\alpha$.
We define $\Pi_\alpha$ as an implicit resolution derivation 
consisting of clauses $C_0, \dots, C_{2^{n^k}{-}1}$. Let $\Delta$ (for ``dummy'') 
be the first clause of $\phi$. For each~$y \in [2^{n^k}]$, if $y \in D_\alpha$ we
set $C_y$ to be the empty clause and derive it by repetition from $C_{N_\alpha(y)}$,
which is an earlier copy of the empty clause.
If $y \notin D_\alpha$ we set $C_y$ to be a copy of $\Delta$, and derive it
by repetition from the original $\Delta$ in $\phi$. 
The final line $C_{2^{n^k}-1}$ is the empty clause and we derive $C_\alpha$ from it
by weakening.
\qed

An interesting case of the construction is 
when there is a family $\phi_n$ of contradictions which are not provably 
contradictions in $S^1_2$, but which nevertheless have quasipolynomial-size refutations
in $G_1$. Then we can typically find a model of $S^1_2$ containing a member $\phi$
of this family together with both a $G_1$ refutation of $\phi$ and at least one assignment~$\alpha$
satisfying~$\phi$. The construction
tells us how to define an exponential-size resolution refutation of~$\phi$ even in such a model.
This does not lead to any inconsistency, since $S^1_2$ cannot prove the soundness
of such a resolution refutation.

%%%%%%%%%%%%%%%%%%%%%%%%%%%%%%%%%%%%%%%%%%%%
\subsection{A generalization to $G_i$ }
%%%%%%%%%%%%%%%%%%%%%%%%%%%%%%%%%%%%%%%%%%%%

In the previous section we showed that there is a polynomial time
function~$f$ which, provably in~$S^1_2$, given a CNF $\phi$ and a $G_1$ refutation of $\phi$ as parameters,
defines an exponential-sized resolution refutation of~$\phi$. 
The simulation of $G_1$ by $\impres(\Res)$ follows by Corollary~\ref{cor:s12_nimpres}.
We would like to generalize this to~$G_{i+1}$ for $i \ge 1$. 
In this direction we can show the following:

\begin{lemma}\label{lem:G_k+1_simulation}
There is a $\PP^{\Sigma^p_i}$
function~$M$ which, provably in~$S^1_2$, given a CNF~$\phi$ and a $G_{i+1}$ refutation of $\phi$ as parameters,
defines an exponential-sized resolution refutation of~$\phi$.
\end{lemma}

In other words, in $S^1_2$ we can convert $G_{i+1}$ proofs into $\PP^{\Sigma^p_i}$-uniform exponential-sized resolution refutations. In the other direction, it is easy to see that $T^{i+1}_2$ proves the soundness of such
 refutations. 
Ideally, we would like to combine these two facts to conclude something like
``$G_{i+1}$ is polynomially equivalent to $\Sigma^p_i$-$\impres(\Res)$''
where $\Sigma^p_i$-$\impres(\Res)$ is a notional proof system
in which a refutation of $\phi$ is a pair $(\Pi, C)$ of a $\Sigma^p_i$ circuit $C$,
that is, a circuit equipped with gates for a $\Sigma^p_i$ oracle, and
a resolution proof~$\Pi$ of the statement that~$C$ encodes the structure of an exponential-sized
resolution refutation of~$\phi$.
We do not draw this conclusion because we do not know how to naturally define such 
a system  $\Sigma^p_i$-$\impres(\Res)$. In particular, it seems that it would
require some machinery to allow resolution to reason naturally about $\Sigma^p_i$ circuits.

\bigskip

\noindent \emph{Proof of Lemma~\ref{lem:G_k+1_simulation}.}
In the previous proof we used that,
because the soundness of $G_1$ is a consequence of $T^1_2$, we know that provably in $S^1_2$,
given a CNF with simultaneously a satisfying assignment and a $G_1$ refutation,
we can produce a PLS instance with no solution. 
Here we will use that, because the soundness of
 $G_{i+1}$ is a consequence of $T^{i+1}_2$, we know that provably in $S^1_2$,
given a CNF with simultaneously a satisfying assignment and a $G_{i+1}$ refutation,
we can produce an instance of the \emph{$(i{+}1)$-turn game induction principle} 
with no solution.
Here the $(i{+}1)$-turn game induction principle, or~$\mathrm{GI}_{i+1}$, 
is a TFNP problem characterizing the $\forall \Sigma^b_1$
consequences of $T^{i+1}_2$~\cite{skelley-thapen}.

An instance of $\mathrm{GI}_{i+1}$ consists of a sequence $H_0, \dots, H_{a-1}$ of $i$-turn games,
where~$a$ is a binary, that is, exponential, size parameter.
Each game is an $i$-ary relation
$H_j(y_1, \dots, y_i)$ in which the moves $y_1, \dots, y_i$ are elements of $[a]$,
with player $A$ making the odd-numbered and player $B$ the even-numbered moves;
they are such that $B$ wins every play of $H_0$, and $A$ wins every play of $H_{a-1}$,
and they come with a suite of purported ``game reductions'' 
which are polynomial-time functions that convert a winning strategy for $B$ in $H_j$
into one in $H_{j+1}$.
Let us write $\theta_j (z)$ for the $\Sigma^b_i$ relation
``there is a winning strategy for $B$ in game $H_j$, if $A$'s first move is~$z$'',
or in short ``$z$ is a losing first move for $A$ in $H_j$''.

By the characterization of $\forall \Sigma^b_1 (T^{i+1}_2)$ in~\cite{skelley-thapen}, 
there is a $\mathrm{GI}_{i+1}$ problem $Q_w$
 with parameter~$w$
such that, provably in~$S^1_2$, if $w$ is a triple $(\phi, \pi, \alpha)$
where
\begin{itemize}
\item
$\phi$ is a CNF in variables $x_1, \dots, x_n$
\item
$\pi$ is a $G_{i+1}$ refutation of $\phi$, and 
\item
$\alpha$ is an assignment satisfying $\phi$
\end{itemize}
then $Q_w$ has no solution. Here ``$Q_w$ has no solution''
means that all the game reductions work correctly;
in particular, we have polynomial-time functions $f_0, \dots, f_{a-2}$
among the reductions
such that $\forall j \forall z , \ \theta_j(z) \rightarrow \theta_{j+1}(f_j(z))$,
meaning that if $z$ is a losing first move for $A$ in $H_j$, then 
$f_j(z)$ is a losing first move for $A$ in $H_{j+1}$.
We may assume that the size parameter~$a$ of the instance
is a function of the number~$n$ of variables in $\phi$ and in particular is independent of~$\alpha$.

We now imitate the previous proof. Again we define a resolution 
refutation of $\phi$ in two parts, where the upper part contains, for each total assignment $\alpha$,
a derivation $\Pi_\alpha$
of $C_\alpha$ from $\phi$, where $C_\alpha$ is the clause expressing the negation of $\alpha$.
The lower part simply resolves all the clauses $C_\alpha$ together in a binary tree to derive the empty clause.

As before, if $\alpha$ falsifies $\phi$ then we set $\Pi_\alpha$ to be a derivation of $C_\alpha$
from the first clause of $\phi$ which is false in~$\alpha$, by a single weakening step.
If $\alpha$ satisfies~$\phi$, then we define $\Pi_\alpha$ to be the following resolution derivation of the empty clause,
from which we can then derive $C_\alpha$ by weakening.
Let $\Delta$ be the first clause of $\phi$.
The underlying graph of $\Pi_\alpha$ consists of all pairs $j,z  \in [a{-}1] \times [a]$ in reverse order, that is, with larger values of~$j$ and~$z$ occurring earlier in the proof.
Node $(j,z)$ is assigned the empty clause $\bot$ if $\theta_j(z)$ holds,
and is otherwise assigned the clause~$\Delta$. Then for each instance of $\bot$ in the proof
occurring at a node $(j,z)$, we can derive it by repetition from the earlier node
$(j+1, f_j(z))$, which will also have been assigned $\bot$ since~$\theta_{j+1}(f_j(z))$.
The final node $(0,0)$ is assigned~$\bot$ since B wins every play of $H_0$
so $\theta_0(0)$ must hold.
\qed

\paragraph{Acknowledgment.} 
We used an LLM to check the manuscript at a final stage of editing, to find typos and grammatical errors.

%%%%%%%%%%%%%%%%%%%%%%%%%%%%%%%%%%%%%%%%%%%%%%%%%%%%%%%
%%%%%%%%%%%%%%%%%%%%%%%%%%%%%%%%%%%%%%%%%%%%%%%%%%%%%%%

\end{document}